\pdfoutput=1
\RequirePackage{ifpdf}
\ifpdf 
\documentclass[pdftex]{sigma}
\else
\documentclass{sigma}
\fi

\usepackage{mathrsfs}

\newcommand{\RR}{\mathbb R}
\newcommand{\CC}{\mathbb C}
\newcommand{\ZZ}{\mathbb Z}
\newcommand{\dd}{{\rm d}}
\DeclareMathOperator*{\slim}{s-lim}
\newcommand{\Span}[1]{\operatorname{Span} \{#1 \}}

\numberwithin{equation}{section}
\newtheorem{Theorem}{Theorem}[section]
 { \theoremstyle{definition}
\newtheorem{Definition}[Theorem]{Definition} }

\begin{document}


\newcommand{\arXivNumber}{1612.05622}

\renewcommand{\PaperNumber}{039}

\FirstPageHeading

\ShortArticleName{Factor Ordering and Path Integral Measure for Quantum Gravity in (1+1) Dimensions}

\ArticleName{Factor Ordering and Path Integral Measure\\ for Quantum Gravity in (1+1) Dimensions}

\Author{John HAGA and Rachel Lash MAITRA}
\AuthorNameForHeading{J.~Haga and R.L.~Maitra}
\Address{Department of Applied Mathematics, Wentworth Institute of Technology,\\ 550 Huntington Ave., Boston MA 02115, USA}
\Email{\href{mailto:hagaj@wit.edu}{hagaj@wit.edu}, \href{mailto:maitrar@wit.edu}{maitrar@wit.edu}}

\ArticleDates{Received December 19, 2016, in f\/inal form June 01, 2017; Published online June 07, 2017}

\Abstract{We develop a mathematically rigorous path integral representation of the time evolution operator for a model of (1+1) quantum gravity that incorporates factor ordering ambiguity. In obtaining a suitable integral kernel for the time-evolution operator, one requires that the corresponding Hamiltonian is self-adjoint; this issue is subtle for a particular category of factor orderings. We identify and parametrize a complete set of self-adjoint extensions and provide a canonical description of these extensions in terms of boundary conditions. Moreover, we use Trotter-type product formulae to construct path-integral representations of time evolution.}

\Keywords{factor ordering in quantum gravity; path integrals in quantum gravity; singularity avoidance in quantum gravity; quantization on a half-line}

\Classification{81V17; 81S40; 83C80}

\section{Introduction}
In general relativity, the conf\/iguration space for the gravitational f\/ield presents two formidable challenges to quantization: f\/irst, it is curved, and second, it has a boundary. In terms of Arnowitt--Deser--Misner (ADM) metric variables, the curvature of the conf\/iguration space of 3-metrics $h_{ab}$ on a spacelike slice is evident in the DeWitt supermetric
\begin{gather}
G_{abcd} = \frac{1}{2 \sqrt{h}} ( h_{ac} h_{bd} + h_{ad} h_{bc} - h_{ab} h_{cd}),
\label{supermetric}
\end{gather}
which appears in the kinetic term of the Hamiltonian constraint
\begin{gather}\label{HamConstraint}
\mathcal{H} = G_{abcd} \pi^{ab} \pi^{cd} - \sqrt{h} \big({}^{(3)}R - 2 \Lambda \big)
\end{gather}
generating time-reparametrization invariance. Here $\pi^{ab} = \sqrt{h} ( h^{ab} K - K^{ab})$ is the momentum canonically conjugate to $h_{ab}$, def\/ined in terms of the extrinsic curvature $K_{ab}$ of the spacelike slice. The potential term of the Hamiltonian constraint~\eqref{HamConstraint} involves the scalar curvature ${}^{(3)}R$ of $h_{ab}$ and a cosmological constant $\Lambda$. The dif\/f\/iculty for quantization engendered by the DeWitt supermetric lies in the ambiguity as to how noncommuting operators~$\hat{\pi}^{ab}$ and~$\hat{h}_{ab}$ are to be ordered in the canonically quantized Hamiltonian constraint, i.e., the Wheeler--DeWitt equation $\hat{\mathcal{H}} \Psi = 0$.

The second problem for quantization, that of the boundary in conf\/iguration space, results from the requirement that spatial 3-metrics be positive def\/inite. In general, conf\/iguration spaces with boundary require the imposition of boundary conditions on quantum states in order for the Hamiltonian operator to be self-adjoint. In the case of quantum gravity, such boundary conditions take on added interest because of their signif\/icance to the quantum behavior of the universe close to the Big Bang/Big Crunch singularity.

The ambiguities of operator ordering and boundary condition are not special to canonical quantization, but manifest themselves in a path integral approach through the choice in def\/inition of a path integral measure and contour of paths to integrate over. The link is most readily seen by considering a (formal) path integral expression for a wave functional
\begin{gather}\label{WavefunctionUniverse}
\Psi [ h_{ab} ] = \mathcal{N} \int_{\mathcal{C}} e^{-I_E [g_{\alpha \beta}]} \mathcal{D} [ g_{\alpha \beta}]
\end{gather}
as in the Hartle--Hawking wave function of the universe \cite{HartleHawking1983}, where $\mathcal{N}$ represents a normalization constant, $I_E[g]$ is the Euclidean signature action evaluated on a trajectory through the superspace of metric conf\/igurations, and $\mathcal{C}$ represents a suitable contour of 4-metrics $g_{\alpha \beta}$ bounded by the 3-metric $h_{ab}$ at time $t=0$. Since the wave functional \eqref{WavefunctionUniverse} must satisfy a Wheeler--DeWitt equation for given factor ordering together with a boundary condition on superspace, and since its only variability is in the def\/inition of the path integral measure $\mathcal{D} [ g_{\alpha \beta} ]$ and contour~$\mathcal{C}$, these choices must be made jointly.

Various proposals exist to motivate the choice of particular orderings or measures on physical grounds, primarily supported by the rationale that classically reparametrization-invariant theories should retain coordinate invariance upon quantization~\cite{Anderson2009,DeWitt1957, Halliwell1988, Moss1988}. Reasoning along these lines has led to the arguments in favor of a Laplace--Beltrami ordering in which the kinetic term of the Hamiltonian is promoted to a covariant Laplace--Beltrami operator on conf\/iguration space. More generally, a multiple of the Ricci curvature of conf\/iguration space may be added without breaking reparametrization invariance; such an ordering is known as a conformal ordering.

An important means of illuminating the question of factor ordering is to investigate the explicit dependence of the path integral measure on the choice of ordering. The subtlety of this correspondence stems in part from the fact that rigorous def\/inition of a path integral measure for a quantum f\/ield theory is intrinsically dif\/f\/icult, since one integrates over paths through a~space of f\/ield conf\/igurations. In addition, two complications arise in general relativity: (1)~the DeWitt supermetric~\eqref{supermetric} has signature $(-+++++)$ (so that one mode always has an opposite sign in the kinetic term), and (2)~the conf\/iguration space of 3-metrics~$g_{ab}$ is restricted by the requirement of positive def\/initeness.

In this paper, we examine the correspondence between factor orderings and path integral measures in the model problem of (1+1) gravity. As shown by Nakayama in~\cite{Nakayama1993}, a~reduction from f\/ield theory to a~mechanical system is possible through gauge f\/ixing, with no imposition of symmetries. This model facilitates a direct consideration of the boundary problem in conf\/iguration space, absent the indef\/inite signature in the kinetic term.

We introduce a 1-parameter family of factor orderings and demonstrate that this is equivalent to the addition of a singular quantum potential in the Hamiltonian. The mathematically rigorous construction of the Feynman path integral involving a singular potential was formulated by Nelson in \cite{Nelson}. Orderings which yield a potential well or suf\/f\/iciently shallow barrier necessitate a boundary condition on quantum states to determine a time evolution operator; a suf\/f\/iciently steep positive quantum potential obviates the need for any boundary condition on wave functions (resulting in a unique time evolution operator).

\looseness=1 Orderings yielding a positive barrier, in concert with a Dirichlet boundary condition on quantum states, permit application of Trotter-type product formulae to express the time evolution operator as a path integral. In the critical case between a potential well and a barrier, path integral representations of time evolution under each possible boundary condition may be constructed for those orderings which yield no added quantum potential. In particular, the Laplace--Beltrami ordering (which here coincides with the conformal ordering) belongs to this class. For orderings yielding a well near the singularity, the product formulae under consideration do not apply, and the question of path integral representation remains for future investigation.

We state some supporting mathematical results in Section \ref{mathback}. In Section~\ref{action} a brief exposition of the model for (1+1) gravity is presented. We canonically quantize the model and introduce factor orderings in Section~\ref{schroquant}. Section~\ref{states} contains a rigorous examination of the domains of quantum states corresponding to each factor ordering in the 1-parameter family, and we derive path-integral representations of time evolution in Section~\ref{FK}.

\section{Mathematical background}\label{mathback}
In this section we state some fundamental results used in subsequent sections. We take $\mathscr H$ to be a separable Hilbert space with inner product $\langle \cdot,\cdot\rangle$. For detailed development of the functional calculus for self-adjoint operators on $\mathscr H$, we direct the reader to \cite{reedsimon1, reedsimon2}; however we state some important def\/initions and theorems for completeness.

\begin{Theorem}[von Neumann]\label{EssentialCriterion}
Let $A$ be a symmetric operator on a Hilbert space $\mathscr H$. Then the following are equivalent:
\begin{enumerate}\itemsep=0pt
\item[$(a)$] $A$ is essentially self-adjoint.
\item[$(b)$] $\operatorname{Ker}(A^*\pm i)=\{0\}$.
\item[$(c)$] $\operatorname{Ran}(A\pm i)$ is dense in $\mathscr H$.
\end{enumerate}
\end{Theorem}

Every symmetric operator $A$ has a closed extension, as $A^*$ is always closed for densely def\/ined operators. However there may be multiple closed extensions of $A$ as we note below.

\begin{Definition}
Suppose that $A$ is a symmetric operator. Let
\begin{gather*}
\mathscr K_\pm = \operatorname{Ker} (i\mp A^*)= \operatorname{Ran}(\pm i+A)^\perp.
\end{gather*}
$\mathscr K_+$ and $\mathscr K_-$ are the \emph{deficiency subspaces} of $A$, and the pair of numbers $n_+$ and $n_-$ given by $n_+(A)=\dim(\mathscr K_-)$ and $n_-(A)=\dim(\mathscr K_+)$ are the \emph{deficiency indices} of~$A$.
\end{Definition}

As famously proven by von Neumann, the self-adjoint extensions of a closed symmetric ope\-ra\-tor bijectively correspond to partial isometries between the def\/iciency subspaces. In particular, considering operators of the form $\hat H=-\hbar^2\frac{\dd^2}{\dd x^2}+V(x)$, we make use of a theorem due to H.~Weyl.
\begin{Definition}
The potential $V(x)$ is in the \emph{limit circle case} at inf\/inity (resp.\ at zero) if for some $\lambda$, all solutions of \begin{gather*}-\varphi''(x)+V(x)\varphi(x)=\lambda \varphi(x)\end{gather*} are square integrable at inf\/inity (resp.\ at zero). If $V(x)$ is not in the limit circle case at inf\/inity (resp.\ at zero), it is said to be in the \emph{limit point case}.
\end{Definition}
\begin{Theorem}[Weyl's limit point-limit circle criterion] \label{weylcriterion} Let $V(x)$ be a continuous real-valued function on $(0,\infty)$. Then $\hat H=-\hbar^2 \frac{\dd^2}{\dd x^2}+V(x)$ is essentially self-adjoint on $C_0^\infty (0,\infty)$ if and only if $V(x)$ is in the limit point case at both zero and infinity.
\end{Theorem}

We seek to characterize the self-adjoint extensions of the Hamiltonian in terms of explicit boundary conditions. The following result appears in \cite{Gorbachukchuk}:

\begin{Theorem}\label{boundaryspacetheorem}
If $A$ is a symmetric operator on $\mathscr H$ with deficiency indices $(n,n)$, $n\leq \infty$, if there is an $n$-dimensional Hilbert space $\mathscr H_b$ $($the ``boundary space''$)$, and if there exist two linear maps $\Gamma_1,\Gamma_2\colon D(A^*)\to\mathscr H_b$ such that for all $f_1,f_2\in D(A^*)$, \begin{gather*}\langle A^* f_1,f_2 \rangle-\langle f_1,A^* f_2 \rangle=\langle \Gamma_1 f_1,\Gamma_2 f_2\rangle_b - \langle \Gamma_2f_1,\Gamma_1f_2\rangle_b,\end{gather*} and for any two $\Psi_1,\Psi_2\in \mathscr H_b$ there exists $g\in D(A^*)$ satisfying \begin{gather*}\Gamma_ig=\Psi_i, \qquad i=1,2,
\end{gather*} then there is a one-to-one correspondence between unitary maps $U\in \mathscr U(\mathscr H_b)$ and the self-adjoint extensions of $A$, where a map $U$ describes the domain of the corresponding extension $A_U$ as \begin{gather*}D(A_U)=\big\{\psi\in D(A^*)\,\big|\,(U-1_b)\Gamma_1\psi+i(U+1_b)\Gamma_2\psi=0\big\}.\end{gather*}
\end{Theorem}

The existence of the boundary space for symmetric $A$ is proven in \cite{Gorbachukchuk}. We apply this theory to the Hamiltonians arising from various factor orderings in Section \ref{states}.

\section{Action for (1+1)-dimensional gravity}\label{action}

Gravity in one space and one time dimension has been considered both as a model problem for full (3+1) geometric theories of gravity and for its own importance to membrane theory. Since Einstein gravity in (1+1) has trivial dynamics, one must devise an alternative formulation for geometric gravity.

The model for $(1+1)$ gravity considered here originates in membrane theory (see the seminal paper~\cite{Polyakov1981} by Polyakov), where gravity is introduced through an area-minimizing action which in the conformal gauge $g_{ab} = e^{\phi} \hat{g}_{ab}$ (with respect to the reference metric $\hat{g}_{ab}$) yields the Liouville action
\begin{gather*}
S_{\rm L} [ \phi, \hat{g} ] = \int_D \left[ \frac{1}{4} \hat{g}^{ab} \partial_a \phi \partial_b \phi + \frac{1}{2} \phi R_{\hat{g}} + 4 \Lambda e^{\phi} \right] \dd^2 x.
\end{gather*}
Subsequent work by Verlinde \cite{Verlinde1990} allows this action to be written in an arbitrary gauge.

By working in the (Euclidean signature) proper-time gauge $g_{00} = 1$, $g_{01} = g_{10} = 0$, $g_{11} = \gamma ( x^0, x^1 )$, Nakayama \cite{Nakayama1993} obtains a reduced quantum mechanical form of Verlinde's action for the (1+1) spacetime manifold~$\RR \times S^1$ parametrized by $(x^0,x^1)$, $0 \le x^1 \le \pi$. The Euler--Lagrange equations are solved to yield a factorization of the space-space component of the metric $\gamma ( x^0 , x^1 ) = \gamma_0 ( x^0 ) \gamma_1 ( x^1 )$. Together with the condition
\begin{gather*}
\int_0^{\pi} \sqrt{\gamma_1(x^1)} \dd x^1 = (m+1) \pi, \qquad m=0,1,2,\dots,
\end{gather*}
this implies that the normalized spatial arc length reduces to
\begin{gather*}
l(x^0) = \frac{1}{\pi} \int_0^{\pi} \sqrt{\gamma (x^0,x^1)} \dd x^1 = (m+1) \sqrt{\gamma_0 ( x^0 )}.
\end{gather*}
Consequently, the Hamiltonian becomes
\begin{gather*}
H_{\rm red} \propto \frac{1}{4l(x^0)} \left( \frac{\dd l}{\dd x^0} \right)^2 - \Lambda l
\end{gather*}
(still in the Euclidean signature), meaning that in Lorentzian signature one may adopt the Lagrangian action
\begin{gather*}
S[l] = \int_0^T \left[\frac{1}{4l} \dot l ^2 - \Lambda l\right] \dd x^0
\end{gather*}
with corresponding Hamiltonian
\begin{gather}\label{RedHam}
H = l \Pi_l^2 + \Lambda l,
\end{gather}
where the momentum conjugate to $l$ is given by $\Pi_l = \frac{\dot{l}}{2l}$. We note that since~$l$ is invariant under reparametrizations of the coordinate $x^1$, the reduced action $S$ and the subsequent constructions below are naturally preserved by spatial re-coordinatization.

Note that in \cite{Nakayama1993}, a term proportional to $l^{-1}$ is included in the reduced Hamiltonian to account for Casimir energy. Since factor ordering of the kinetic term of the quantized Hamiltonian will introduce such a term naturally (see Section~\ref{schroquant}), we regard~\eqref{RedHam} as the classical Hamiltonian.

\section{Schr\"{o}dinger quantization}\label{schroquant}

A prototype for sampling the ef\/fects of factor ordering ambiguity has been to consider a family of orderings of the form (e.g., \cite{AmbjornGlaserSatoWatabiki, HartleHawking1983, MaitraConfProc,SteiglHinterleitner})
\begin{gather}\label{ijkHam}
\hat{H}_{(j_1,j_2,j_3)} = -\hbar^2 l^{j_1} \frac{\dd}{\dd l} l^{j_2} \frac{\dd}{\dd l} l^{j_3} + \Lambda l, \qquad {j_1}+{j_2}+{j_3}=1 ,
\end{gather}
where the quantized circumference operator $\hat{l} \colon \Psi \to l \Psi$ and its conjugate momentum operator $\hat{\Pi}_l \colon \Psi \to -i \hbar \frac{d \Psi}{d l}$ are variously composed.

Denoting $J_{\pm}=j_3\pm j_1$, \eqref{ijkHam} is formally self-adjoint with respect to the measure $l^{J_-} \dd l$ on $\RR^+ = ( 0 , \infty )$. Transforming states and observables according to
\begin{gather}
L^2 \big( \RR^+ , \mu(l) \dd l \big) \to L^2 \big( \RR^+ , \nu(l) \dd l \big), \notag \\
\Psi \mapsto \tilde{\Psi} = \left[ \frac{\mu(l)}{\nu(l)} \right]^{1/2} \Psi, \notag \\
\hat{A} \mapsto \tilde{A} = \left[ \frac{\mu(l)}{\nu(l)} \right]^{1/2} A \left[ \frac{\mu(l)}{\nu(l)} \right]^{-1/2} ,\label{transHilbert}
\end{gather}
with $\mu(l) = l^{J_-}$ and $\nu(l) = l^{-1/2}$, \eqref{ijkHam} maps to an operator $\tilde H$ which is formally self-adjoint with respect to a~standard measure $l^{-1/2} \dd l$. Thus we have
\begin{gather*}
\hat{H}_{(j_1,j_2,j_3)}=l^{-m} \tilde H l^{m},
\end{gather*}
where $m=\frac{1}{4} + \frac{J_-}{2}$ and
\begin{gather*}
\tilde H = -\hbar^2 l^{1/2} \frac{\dd}{\dd l} l^{1/2} \frac{\dd}{\dd l} + \frac{\hbar^2}{4} q l^{-1} + \Lambda l, \qquad q = J_+^2 - \frac{1}{4} .
\end{gather*}

Normalization of states and formal self-adjointness of operators is preserved under the transformation~\eqref{transHilbert}. Note that the dif\/ferential operator $ l^{1/2} \frac{\dd}{\dd l} l^{1/2} \frac{\dd}{\dd l}$ is the Laplace--Beltrami ope\-ra\-tor on~$\RR^+$ with respect to the metric $g=(l^{-1})$, and the new ef\/fective potential $\frac{\hbar^2}{4} ql^{-1} + \Lambda l$ dif\/fers from the original potential by the singular quantum potential $\frac{\hbar^2}{4} ql^{-1}$. Thus, the Laplace--Beltrami ordering ($j_1=\frac{1}{2}=j_2$, $j_3=0$) belongs to the class of orderings $|J_+|=\frac{1}{2}$, where the quantum potential vanishes.

In the parameter space $(j_1,j_3)$ of orderings, the direction $J_+$ is physically signif\/icant in dictating the coef\/f\/icient of the quantum potential. In contrast, $J_-$ is insignif\/icant, because if $J_+=J_+'$ for two orderings $(j_1,j_2,j_3)$ and $(j_1',j_2',j_3')$, we have
\begin{gather*}
\hat{H}_{(j_1',j_2',j_3')} = l^{-\alpha} \hat{H}_{(j_1,j_2,j_3)} l^{\alpha}, \qquad \alpha = \frac{1}{2} ( J_- - J_-').
\end{gather*}

A further substitution $y = 2 l^{1/2}$, with $\dd y = l^{-1/2} \dd l$ and $\frac{\dd^2}{\dd y^2}=l^{1/2} \frac{\dd}{\dd l} l^{1/2} \frac{\dd}{\dd l}$, transforms the curved coordinate $l$ to a f\/lat one, yielding
\begin{gather}
\tilde{H} = - \hbar^2 \frac{\dd^2}{\dd y^2} + \hbar^2 q y^{-2} + \frac{\Lambda}{4} y^2 .\label{flat_ham}
\end{gather}
In what follows, we denote the kinetic and potential terms of $\tilde{H}$ as
\begin{gather*}
\tilde T = - \hbar^2 \Delta = - \hbar^2 \frac{\dd^2}{\dd y^2}, \qquad \tilde{V} = \hbar^2 q y^{-2} + \frac{\Lambda}{4} y^2 .
\end{gather*}
The Hamiltonian $\tilde H$, including the singular term $\hbar q y^{-2}$, resembles radial Hamiltonians arising in nuclear physics; see, e.g.,~\cite{Kleinert}.

\section[Self-adjoint extensions of $\tilde{H}$]{Self-adjoint extensions of $\boldsymbol{\tilde{H}}$}\label{states}
In def\/ining a time evolution operator, the Hamiltonian (as the inf\/initesimal generator) must be self-adjoint (cf.\ Stone's theorem in~\cite{lax}). Indeed, if one chooses $D(\tilde H)=C_0^\infty(\RR^+)$, standard analysis (see~\cite[Chapter 4]{gitmantyutin}) yields that
\begin{gather}
D(\tilde H^*) = \big\{ \psi \, | \, \psi , \psi' \; \mbox{a.c. in}\; (0,\infty); \; \psi , \tilde{H} \psi \in L^2(\RR^+) \big\},\label{naturaldomain}
\end{gather}
the so-called \emph{natural domain} of the dif\/ferential operation $\check H=- \hbar^2 \frac{\dd^2}{\dd y^2} + \hbar^2 q y^{-2} + \frac{\Lambda}{4} y^2$. In Section~\ref{deficiencysubspaces} we show that $\tilde H$ has a unique self-adjoint extension for a large class of orderings; in Sections~\ref{Boundaryspaces} and~\ref{Asymptotics} we parametrize and characterize the self-adjoint extensions corresponding to those orderings admiting multiple extensions.

\subsection{Def\/iciency subspaces}\label{deficiencysubspaces}

We proceed by determining def\/iciency subspaces for $\tilde H$. To compute the def\/iciency indices of~$\tilde H$ we search for solutions to
\begin{gather}
(\tilde H^*-(\pm i)I)\psi(y)=0 ;\label{defindeq}
\end{gather}
the solution space will be denoted by $\mathscr K^{\pm}$ (respectively). To transform \eqref{defindeq}, we def\/ine the variables $z$ and $\varphi$ as follows
\begin{gather*}z=\frac{\sqrt \Lambda }{2\hbar}y^2\qquad \textrm{and} \qquad \psi(y)=z^\beta e^{-z/2}\varphi(z),
\end{gather*}
where $\beta=\frac{1}{4}+\frac{|J_+|}{2}$. From this~\eqref{defindeq} becomes
\begin{gather}\label{subdefindeq}
z\frac{\dd^2\varphi}{\dd z^2}+(1+ |J_+|-z)\frac{\dd \varphi}{\dd z}-\left(\frac{1+|J_+|\pm i/(\hbar\sqrt \Lambda)}{2}\right)\varphi(z)=0.
\end{gather}
Def\/ining $\alpha =(1+|J_+|\pm i/(\hbar\sqrt \Lambda))/2$ and $\gamma = 1 + |J_+|$, \eqref{subdefindeq} is recognizable as the conf\/luent hypergeometric equation
\begin{gather}\label{chgeqn1}
z \frac{\dd^2\varphi}{\dd z^2}+(\gamma - z)\frac{\dd \varphi}{\dd z}-\alpha \varphi=0 .
\end{gather} Solutions to \eqref{chgeqn1} are conf\/luent hypergeometric functions of the f\/irst and second kind (respectively)~\cite{lebedev}:
\begin{gather*}
\varphi_{1}(z)=\Phi(\alpha,\gamma;z) :=\sum_{k=0}^\infty \frac{(\alpha)_k}{(\gamma)_k}\frac{z^k}{k!},\\
\varphi_{2}(z)=\Psi(\alpha,\gamma;z) :=\frac{\Gamma(1-\gamma)}{\Gamma(1+\alpha-\gamma)}\Phi(\alpha,\gamma;z)+\frac{\Gamma(\gamma-1)}
{\Gamma(\alpha)}z^{1-\gamma}\Phi(1+\alpha-\gamma,2-\gamma;z),
\end{gather*}
where $|z|<\infty$, $|\arg(z)|<\pi$, $\gamma\not=0, -1, -2, \dots$, and the Pochhammer symbol $(a)_k = a(a+1)\cdots (a+k-1)$ for $a\in \CC$ and integer $k\geq 0$. We obtain the following solutions to~\eqref{defindeq}:
\begin{gather*}
\psi_1^\pm(y) =z^\beta e^{-z/2}\sum_{k=0}^\infty \left(\frac{1+|J_+|}{2}\pm\frac{i}{2\hbar \sqrt \Lambda}\right)_k \frac{z^k}{k!(1+|J_+|)_k},\\ 
\psi_2^\pm(y) =\frac{\Gamma(-|J_+|)}{\Gamma\left(\frac{1-|J_+|}{2}\pm \frac{i}{2\hbar \sqrt \Lambda} \right)}z^\beta e^{-z/2}\sum_{k=0}^\infty \left(\frac{1+|J_+|}{2}\pm\frac{i}{2\hbar \sqrt \Lambda}\right)_k \frac{z^k}{k!(1+|J_+|)_k}\nonumber\\
\hphantom{\psi_2^\pm(y)=}{}
+\frac{\Gamma(|J_+|)}{\Gamma\left(\frac{1+|J_+|}{2}\pm \frac{i}{2\hbar \sqrt \Lambda} \right)}z^{\beta-|J_+|} e^{-z/2}\sum_{k=0}^\infty \left(\frac{1-|J_+|}{2}\pm\frac{i}{2\hbar \sqrt \Lambda}\right)_k \frac{z^k}{k!(1-|J_+|)_k} .
\end{gather*}

With the intent of applying Theorem \ref{weylcriterion} we investigate asymptotic behavior of these solutions for small and large~$y$. As $y\to 0^+$,
\begin{gather*}
\psi_1^\pm(y) =O\big( y^{2\beta}\big)=O\big(y^{\frac{1}{2}+|J_+|}\big),\\
\psi_2^\pm(y) =O\big(y^{2\beta}\big) + O\big(y^{2\beta-2|J_+|}\big)=O\big(y^{\frac{1}{2}+|J_+|}\big)+O\big(y^{\frac{1}{2}-|J_+|}\big)=O\big(y^{\frac{1}{2}-|J_+|}\big).
\end{gather*}
Thus $\psi_1^\pm \in L^2$ near $0$, and $\psi_2^\pm \in L^2$ near $0$ precisely when $|J_+|<1$. As in~\cite{lebedev}, for $z\to \infty$:
\begin{gather*}
\Phi(\alpha,\gamma;z)=O\big(e^zz^{-(\gamma-\alpha)}\big) , \qquad \Psi(\alpha,\gamma;z)=O\big(z^{-\alpha}\big),
\end{gather*}
and as $y\to \infty$,
\begin{gather*}
\psi_1^\pm(y)=O\big(e^{\frac{\sqrt \Lambda}{4\hbar}y^2}y^{-\frac{1}{2}\mp \frac{i}{\hbar \sqrt \Lambda}}\big) , \qquad \psi_2^\pm(y)=O\big(e^{-\frac{\sqrt \Lambda}{4\hbar}y^2}y^{\frac{3}{2}+2|J_+|\pm\frac{i}{\hbar \sqrt \Lambda}}\big) .
\end{gather*}
Thus $\psi^\pm_1\not\in L^2$ near $\infty$ and $\psi^\pm_2\in L^2$ near $\infty$ for all factor orderings. Consequently for $|J_+| \geq 1$, $\mathscr K^{\pm} = \{0\}$ and $n_+(\tilde{H}) = 0 = n_-(\tilde{H})$. For $|J_+|<1$, $\mathscr K^{\pm} = \Span{ \psi_2^\pm }$ and $n_+(\tilde{H}) = 1 = n_-(\tilde{H})$. These analyses support the following result related to the Hamiltonian~\eqref{flat_ham}:

\begin{Theorem}\label{ess_sa_cutoff}
The operator $\tilde H=-\hbar^2\frac{\dd^2}{\dd y^2}+\hbar^2\big(|J_+|^2-\frac{1}{4}\big)\frac{1}{y^2}+\frac{\Lambda}{4}y^2$ is essentially self-adjoint on $C_0^\infty (0,\infty)$ iff $|J_+|\geq 1$.
\end{Theorem}
\begin{proof}
We have seen that $\psi^\pm_2(y)$ is square integrable near zero if\/f $|J_+|<1$. Thus, $V(y)=\hbar^2\big(|J_+|^2-\frac{1}{4}\big)\frac{1}{y^2}+\frac{\Lambda}{4}y^2$ is in limit point case at zero if\/f $|J_+|\geq 1$. Since $\psi^\pm_1(y)$ is never square integrable at inf\/inity, $V(y)$ is always in limit point case at inf\/inity. Applying Theorem~\ref{weylcriterion} now establishes the desired result.
\end{proof}

\subsection{Self-adjoint extensions in terms of boundary spaces}\label{Boundaryspaces}

For the case when $|J_+|<1$, we seek to characterize and classify the self-adjoint extensions of $\tilde H$ in terms of boundary conditions as in Theorem~\ref{boundaryspacetheorem}. Our method is similar to that presented in~\cite{fulop}. Suppose that $
\tilde{\mathscr H_b}=\CC $, and $\varphi^{(1)}$, $\varphi^{(2)}$ are solutions to $\tilde H \varphi=0$ with Wronskian $W[\varphi^{(1)},\varphi^{(2)}]=1$ (referred to as \emph{reference modes}). Let $\Gamma_i\colon D(\tilde H^*)\to \CC$ be def\/ined by
\begin{gather*}
\Gamma_i \psi=\lim_{y\to 0^+}W\big[\varphi^{(i)},\psi\big], \qquad \textrm{for} \quad i=1,2 .
\end{gather*}

A straightforward computation establishes that for all $\psi_1,\psi_2\in D(\tilde H^*)$, \begin{gather*}\big\langle \tilde H^* \psi_1, \psi_2\big\rangle-\big\langle \psi_1,\tilde H^* \psi_2 \big\rangle=\langle \Gamma_1 \psi_1,\Gamma_2 \psi_2 \rangle_\CC-\langle \Gamma_2 \psi_1,\Gamma_1 \psi_2 \rangle_\CC ,\end{gather*} and for any $\Psi_1,\Psi_2\in \CC$, one may construct $g\in D(\tilde H^*)$ (by choosing $g_{1,2}\sim C_{1,2}y^{\pm2\nu+\frac{1}{2}}$ and taking $g=g_1+g_2$) satisfying that $\Gamma_i g=\Psi_i$, $i=1,2$.

To determine an explicit formula for the reference modes, we write
\begin{gather*}
\tilde H \varphi = -\hbar^2\varphi''(y)+\left(\frac{\hbar^2q}{y^{2}}+\frac{\Lambda y^2}{4}\right)\varphi(y)=0,
\end{gather*}
where, as before, $q=J_+^2-\frac{1}{4}$. Making the substitutions $w(y) =\frac{\sqrt{\Lambda}}{4 \hbar} y^2$ and $\varphi(y)=w^{1/4} u(w)$ we obtain
\begin{gather*}
w^2u''(w)+wu'(w)-\left(w^2+\left(\frac{J_+}{2}\right)^2\right)u(w)=0.
\end{gather*}
This is a modif\/ied Bessel equation; letting $\nu = \frac{| J_+ |}{2}$, solutions to this equation are given by the modif\/ied Bessel functions of the f\/irst ($I_\nu(w)$) and second ($K_\nu(w)$) kind; their Wronskian satisf\/ies $W[K_\nu(w),I_\nu(w)]=1/w$. Choosing \begin{gather*}
\varphi^{(1)}(y)=\frac{\sqrt y}{\sqrt 2}K_\nu\left(\frac{\sqrt \Lambda}{4\hbar}y^2\right) \qquad \textrm{and} \qquad
\varphi^{(2)}(y)=\frac{\sqrt y}{\sqrt 2}I_\nu\left(\frac{\sqrt \Lambda}{4\hbar}y^2\right),
\end{gather*}
it follows that $W[\varphi^{(1)},\varphi^{(2)}]=1$.

The asymptotic equivalences (for $w\to 0$) of modif\/ied Bessel functions are well-known:
\begin{gather*}
K_\nu (w) \sim \begin{cases} \displaystyle \frac{\pi}{2} \frac{1}{\sin (\nu \pi)} \frac{1}{\Gamma(1-\nu)} \left( \frac{w}{2} \right)^{-\nu} , & \nu \notin \ZZ, \vspace{1mm}\\
\displaystyle - \ln \left( \frac{w}{2} \right) , & \nu = 0,
\end{cases} \qquad
I_\nu (w) \sim \frac{1}{\Gamma (\nu + 1)} \left( \frac{w}{2} \right)^\nu .
\end{gather*}
These yield
\begin{gather}
\label{asymptotpsi1}
\varphi^{(1)}(y) \sim \begin{cases}
\mathscr{C}_1 y^{-2\nu + \frac{1}{2}} , & 0 < |J_+| < 1, \\
-\sqrt 2 y^{1/2} \ln y + k y^{1/2} , & J_+ = 0 ,
\end{cases}
\end{gather}
where
\begin{gather*}
\mathscr{C}_1 = \frac{\pi}{2 \sqrt{2}} \frac{1}{\sin (\nu \pi)} \left( \frac{\sqrt{\Lambda}}{8 \hbar} \right)^{-\nu} \frac{1}{\Gamma ( 1- \nu )}, \qquad
k = - \frac{1}{\sqrt{2} } \ln \left( \frac{\sqrt{\Lambda}}{8 \hbar} \right) ,
\end{gather*}
and
\begin{gather}\label{asymptotpsi2}
\varphi^{(2)} (y) \sim \mathscr{C}_2 y^{2 \nu + \frac{1}{2}} ,
\end{gather}
where
\begin{gather*}
\mathscr{C}_2 = \frac{1}{\sqrt{2}} \frac{1}{\Gamma (\nu + 1)} \left( \frac{\sqrt{\Lambda}}{8 \hbar} \right)^\nu .
\end{gather*}
The corresponding asymptotic behavior of the derivatives is given by
\begin{gather*}
{\varphi^{(1)}}'(y) \sim \begin{cases} \left( -2 \nu + \frac{1}{2} \right) \mathscr{C}_1 y^{-2 \nu - \frac{1}{2}} , & 0 < |J_+| < 1, \\
 -\frac{1}{\sqrt 2} y^{-1/2} \ln y + \left( \frac{k}{2} - \sqrt{2} \right) y^{-1/2} , & J_+=0,
 \end{cases} 
\\
{\varphi^{(2)}}' (y) \sim \left( 2 \nu + \tfrac{1}{2} \right) \mathscr{C}_2 y^{2 \nu - \frac{1}{2}} . 
\end{gather*}
For detailed treatment of modif\/ied Bessel functions and asymptotic behavior we direct the reader to~\cite{lebedev}. Let $\xi\in U(1)$ and $\psi\in D(\tilde H^*)$. Writing $L=\tan(\arg(\xi)/2)$ we have, for $0<|J_+|<1$:
\begin{gather}
(\xi-1) \Gamma_1\psi+i(\xi+1)\Gamma_2\psi\nonumber\\
\qquad{}= \lim_{y\to 0^+}\psi(y)\big[L (2\beta-1 )\mathscr C_1y^{-2\beta}-2\beta\mathscr C_2 y^{2\beta-1}\big]+\psi'(y)\big[L\mathscr C_1y^{1-2\beta}+\mathscr C_2 y^{2\beta}\big],\label{boundaryextension1}
\end{gather}
where as before, $\beta=\frac{1}{4}+\frac{|J_+|}{2}$. For $|J_+|=0$ we have
\begin{gather}
(\xi -1)\Gamma_1\psi+i(\xi+1)\Gamma_2\psi\nonumber\\
\qquad{} =\lim_{y\to 0^+}\!\frac{\psi(y)}{\sqrt y}\big[L\big( \sqrt 2\ln(y)-\tfrac{k}{2}+\sqrt 2\big)\!-\tfrac{1}{2}\mathscr C_2 \big]+\psi'(y)\sqrt y\big[L\big({-}\sqrt 2\ln(y)+k\big)+\mathscr C_2 \big].\!\!\!\!\!\!\!\label{boundaryextension2}
\end{gather}
The self-adjoint extensions of $\tilde H$ are parametrized by $L\in (-\infty,\infty]$ and may thus be characte\-ri\-zed as the restriction of $\tilde H^*$ to functions in $D(\tilde H^*)$ such that the limits in~\eqref{boundaryextension1} and~\eqref{boundaryextension2} (for respective values of~$|J_+|$) are zero.

\subsection{Self-adjoint extensions in terms of asymptotics} \label{Asymptotics}

An alternate approach to characterizing self-adjoint domains for $\tilde H$ involves describing the asymptotic behavior of functions in the natural domain~\eqref{naturaldomain} of $\check H$ and imposing restrictions which force $\lim\limits_{\varepsilon \to 0^+}\lim\limits_{L \to \infty} ( \overline{\psi'} \psi - \overline{\psi}\psi' ) |_{\varepsilon}^L=0$. Combined with the polarization identity, such a restriction causes the integration-by-parts terms obstructing self-adjointness of $\tilde H$ to vanish (see, e.g., \cite[Section~3.2]{gitmantyutin}).

The asymptotic analysis of the natural domain is accomplished by applying methods similar to those used by Gitman et al.~\cite{GitmanTyutinVoronovPaper, gitmantyutin} for the Calogero problem. The condition $\check H \psi \in L^2(\RR^+)$ ef\/fectively describes each $\psi \in D (\tilde H^*)$ as the solution of some dif\/ferential equation
\begin{gather}\label{domaineq}
\check H \psi = \chi, \qquad \chi \in L^2(\RR^+) .
\end{gather}
The general solution to \eqref{domaineq} can be written in terms of the reference modes $\varphi^{(1)}$, $\varphi^{(2)}$ as
\begin{gather}
\psi = c_1 \varphi^{(1)} + c_2 \varphi^{(2)} + \tilde{\psi} , \qquad c_1 ,c_1 \in \CC, \nonumber\\
\tilde \psi= \frac{1}{\hbar^2} \left[ \varphi^{(1)} \int_{a_2}^y \chi \varphi^{(2)} \dd y' - \varphi^{(2)} \int_{a_1}^y \chi \varphi^{(1)} \dd y' \right] ,
\label{gensol}
\end{gather}
since a straightforward computation verif\/ies that $\tilde \psi$ is a particular solution to \eqref{domaineq}.

Utilizing the asymptotic behavior of the reference modes \eqref{asymptotpsi1}, \eqref{asymptotpsi2} (valid for all values of~$|J_+|$, with minor adjustment to coef\/f\/icients) and applying the Cauchy--Schwarz inequality yields
\begin{gather}
\label{asint}
\int_{a_1}^y \chi \tilde{\psi}_1 \dd y' = \begin{cases} O \big(y^{-|J_+|+1}\big) , & |J_+| \neq 1, \\ O \big(\sqrt{\ln (y)}\big) , & |J_+| = 1,
\end{cases} \qquad
\int_{a_2}^y \chi \tilde{\psi}_2 \dd y' = O \big(y^{|J_+|+1}\big) .
\end{gather}
For $|J_+| \neq 0$, we have that when $|J_+|\not=1$, $\tilde \psi=O(y^{3/2})$ and $\tilde \psi'=O(y^{1/2})$; when $|J_+|=1$, $\tilde \psi=O(y^{3/2}\ln (y))$ and $\tilde \psi'=O(y^{1/2}\ln(y))$.

For $J_+ = 0$, we apply analogous asymptotics (relying upon a series relation between $K_0(w)$ and $I_0(w)$ (see, e.g.,~\cite{DLMF})):
\begin{gather*}
K_0 (w(y)) = - \ln \left( \frac{w(y)}{2} \right) I_0 (w(y)) + \sum_{k=0}^\infty \frac{w(y)^{2k}}{2^k (k!)^2} \Psi(k+1) = -2 \big[ \ln(y) I_0 (w(y)) + R(y) \big] ,
\end{gather*}
where
\begin{gather*}
R(y) = \frac{1}{2} \left[ \ln \left( \frac{\sqrt{\Lambda}}{8 \hbar} \right) I_0 (w(y)) - \sum_{k=0}^\infty \frac{w(y)^{2k}}{2^k (k!)^2} \Psi(k+1) \right],
\end{gather*}
and $\Psi$ is the digamma function $\Psi(x) = \frac{\dd}{\dd x} [ \ln (\Gamma(x))]$. Note that $R(y) = O(1)$ as $y \to 0$. Thus we can write (taking $a_1 = 0 = a_2$)
\begin{gather}
\tilde \psi (y) = - \frac{\sqrt{y}}{\hbar^2} I_0 (w(y)) \left[ \ln y \int_0^y \chi(y') \sqrt{y'} I_0 (z(y')) \dd y' - \int_0^y \chi (y') \sqrt{y'} \ln(y') I_0 (w(y')) \dd y' \right] \nonumber\\
\hphantom{\tilde \psi (y) =}{} - \frac{\sqrt{y}}{\hbar^2} \left[ R(y) \int_0^y \chi(y') \sqrt{y'} I_0 (w(y')) \dd y' - I_0 (w(y)) \int_0^y \chi (y') \sqrt{y'} R(y') \dd y' \right].\label{psi-starJ+0}
\end{gather}
To simplify the asymptotic analysis of \eqref{psi-starJ+0} as $y \to 0$, def\/ine
\begin{gather*}
\Theta (y) = \int_0^y \chi(y') \sqrt{y'} I_0 (w(y')) \dd y', \qquad \Theta ' (y) = \chi(y) \sqrt{y} I_0 (w(y)) ,
\end{gather*}
and use the Cauchy--Schwartz inequality as in \eqref{asint} to obtain $\Theta(y) = O(y)$ as $y \to 0$. We can now use integration by parts to rewrite \eqref{psi-starJ+0} as
\begin{gather}
\tilde \psi(y) = \frac{\sqrt{y}}{\hbar^2} I_0 (w(y)) \left[ \int_0^y \ln (y') \Theta'(y') \dd y' - \ln y \Theta(y) \right] \nonumber\\
\hphantom{\tilde \psi(y) =}{} + \frac{\sqrt{y}}{\hbar^2} \left[ I_0 (w(y)) \int_0^y \chi (y') \sqrt{y'} R(y') \dd y' - R(y) \int_0^y \chi(y') \sqrt{y'} I_0 (w(y')) \dd y' \right] \nonumber\\
\hphantom{\tilde \psi(y)}{} = - \frac{\sqrt{y}}{\hbar^2} I_0 (w(y)) \int_0^y \frac{1}{y'} \Theta (y') \dd y' \nonumber\\
\hphantom{\tilde \psi(y) =}{} + \frac{1}{\hbar^2} \left[ A(y) \int_0^y \chi (y') B(y') \dd y' - B(y) \int_0^y \chi(y') A(y') \dd y' \right] ,\label{rewritewTheta}
\end{gather}
where for simplicity $A(y) = \sqrt{y} I_0(w(y))$ and $B(y) = \sqrt{y} R(y)$. Using the fact that $\Theta (y) = O(y)$ and $R(y) = O(1)$, we can once again apply the Cauchy--Schwartz inequality to derive the asymtoptic estimates
\begin{gather*}
\int_0^y \chi A \dd y' = O (y) = \int_0^y \chi B \dd y' ,
\end{gather*}
so that $\tilde \psi = O(y^{3/2})$. Dif\/ferentiating \eqref{rewritewTheta} and applying similar asymptotic analysis yields $\tilde \psi' = O(y^{1/2})$.

When $|J_+| \geq 1$, $\tilde H$ is essentially self-adjoint. Indeed, in this case, \eqref{gensol} implies $c_1 = 0$ since $\psi \in L^2(\RR^+)$; furthermore, the term $c_2 \varphi^{(2)}$ is absorbed into the asymptotic term corresponding to $\tilde \psi$, permitting no freedom in the choice of self-adjoint domain.

When $\tilde H$ is not essentially self-adjoint (i.e., when $|J_+|<1$), a 1-parameter family of pairs $(c_1, c_2)$ indexes the set of self-adjoint extensions. We determine the self-adjoint domains for~$\tilde{H}$ by requiring that
$\lim\limits_{\varepsilon \to 0^+}\lim\limits_{ L \to \infty} ( \overline{\psi'} \psi - \overline{\psi} \psi' )|_{\varepsilon}^L = 0$ ,
or, since $\lim\limits_{L \to \infty} \psi(L) = 0$ and $\lim\limits_{L \to \infty} \psi'(L) = 0$, that
$\lim\limits_{\varepsilon \to 0} ( \overline{\psi'} \psi - \overline{\psi}\psi' )|_{\varepsilon} = 0$. The asymptotic equivalences for $0<|J_+|<1$
\begin{gather*}
\psi \sim c_1 \mathscr{C}_1 y^{-|J_+|+ \frac{1}{2}} + c_2 \mathscr{C}_2 y^{|J_+| + \frac{1}{2}} + O \big( y^{\frac{3}{2}} \big), \\
\psi' \sim c_1 \mathscr{C}_1 \big({-} |J_+| + \tfrac{1}{2} \big) y^{-|J_+| - \frac{1}{2}} + c_2 \mathscr{C}_2 \big( |J_+|+\tfrac{1}{2} \big)y^{|J_+|-\frac{1}{2}} + O\big(y^{\frac{1}{2}}\big)
\end{gather*}
imply that $\lim\limits_{\varepsilon \to 0} ( \overline{\psi'} \psi - \overline{\psi} \psi' )|_\varepsilon = - 2 |J_+| \mathscr{C}_1 \mathscr{C}_2 ( \overline c_1 c_2 - c_1 \overline c_2 )$, so we must require $\overline c_1 c_2 - c_1 \overline c_2 = 0$ to obtain a domain on which $\tilde{H}$ is self adjoint. For orderings satisfying $J_+=0$, we have instead the asymptotic equivalences
\begin{gather*}
\psi \sim c_1 \big( {-}\sqrt 2 y^{\frac{1}{2}} \ln (y) - ky^{\frac{1}{2}} \big) + c_2 \mathscr{C}_2 y^{\frac{1}{2}} + O\big(y^{\frac{3}{2}}\big), \\
\psi' \sim c_1 \big( {-}\tfrac{1}{\sqrt 2} y^{-\frac{1}{2}} \ln (y) + \big(\tfrac{k}{2} - \sqrt{2} \big) y^{-\frac{1}{2}} \big) + c_2 \frac{1}{2} \mathscr{C}_2 y^{-\frac{1}{2}} + O\big(y^{\frac{1}{2}}\big) ,
\end{gather*}
where as before, $k=-\frac{1}{\sqrt 2}\ln\big(\frac{\sqrt \Lambda}{8\hbar}\big)$. These similarly yield $\lim\limits_{\varepsilon \to 0} ( \overline{\psi'} \psi - \overline{\psi} \psi' ) |_\varepsilon =- \sqrt{2} \mathscr{C}_2 ( \overline{c_1} c_2 - c_1 \overline{c_2} ) / \hbar$. Thus $\overline{c_1} c_2 = c_1 \overline{c_2}$, implying $\overline{c_1} c_2 \in \RR$, and hence $c_2$ is a real multiple of $c_1$, since
\begin{gather*}
\alpha c_1 = (\overline{c_2} c_1) c_1 = ( \overline{c_1} c_2 ) c_1 = ( \overline{c_1} c_1 ) c_2 = \beta c_2 , \qquad \alpha , \beta \in \RR.
\end{gather*}
Note that we can def\/ine $\theta$ such that $\tan \theta = \frac{\beta}{\alpha}$, yielding $c_1 \cos (\theta) = c_2 \sin (\theta)$, and allowing us to write $\psi$ in the form
\begin{gather}
\psi =C \big( \sin (\theta) \varphi^{(1)} + \cos (\theta) \varphi^{(2)} \big) + O\big(y^{\frac{3}{2}}\big) ,\label{asymptotics_theta}
\end{gather}
implying the asymptotics
\begin{gather}
\psi\sim\begin{cases}C\big(\sin(\theta)\mathscr C_1y^{-|J_+|+\frac{1}{2}}+\cos(\theta)\mathscr C_2y^{|J_+|+\frac{1}{2}}\big)+O\big(y^{\frac{3}{2}}\big),&0<|J_+|<1,\\
C\big(\sin(\theta)\big({-}\sqrt 2y^{\frac{1}{2}}\ln(y)+ky^{\frac{1}{2}}\big)
+\cos(\theta)\mathscr C_2y^{\frac{1}{2}}\big)+O\big(y^{\frac{3}{2}}\big),&|J_+|=0 ,
\end{cases}\label{thetasymptot}
\end{gather}
and similarly for $\psi'$. Taking $\theta=\tan^{-1}(L)$ we observe that this parametrization of the self-adjoint extensions for $\tilde H$ may be recast in the form as in \eqref{boundaryextension1} and \eqref{boundaryextension2}.

\section[Formulation of the propagator and path-integral representation of the fundamental solution]{Formulation of the propagator and path-integral\\ representation of the fundamental solution}\label{FK}

In this section, we use extensions of the Trotter product formula due to Kato and to Exner, Neidhardt, and Zagrebnov to derive Feynman--Ka\v{c} formulae (i.e., path integral representations) for the time-evolution operators $e^{-it \tilde H/\hbar}$ and $e^{-t \tilde H/\hbar}$ evolving initial data according to the real- and imaginary-time Schr\"{o}dinger equations, respectively. By factoring the time evolution ope\-ra\-tor $\mathscr{U}_t = e^{-a t \tilde H/\hbar}$ $(a= i,1)$ as the limit of the concatenation of short-time evolution operators $\big( e^{-a t \hat T / \hbar n} e^{- at \hat V / \hbar n} \big)^n$,
we can write a path-integral expression for the propagator (integral kernel) of $e^{-at \tilde H} = e^{-a t(\tilde T + \tilde V)}$ by using the known integral kernel for $e^{- at \tilde T /\hbar n} = e^{a \hbar t \Delta/n}$.

The Trotter product formula and its extensions require that the operators $\tilde T$ and $\tilde V$ be self-adjoint \cite{hallquant}, so the domain of $\tilde T = -\hbar^2 \frac{d^2}{dy^2}$ must correspond to one of the self-adjoint extensions of the Laplacian on the half-line; namely, $\tilde T$ must be taken to be $\tilde{T}_{\beta}$, $\beta \in (-\infty, \infty]$, where
\begin{gather*}
\tilde{T}_\beta = -\hbar^2 \Delta = -\hbar^2 \frac{\dd^2}{\dd y^2}, \\
D(\hat{T}_\beta)= \big\{ \psi \in D^*_{\check{T}} \, | \, \lim_{\varepsilon \to 0} [ \psi(\varepsilon) - \beta \psi' (\varepsilon) ] =0 \big\}
\end{gather*}
(as usual, $\beta = \infty$ denotes the Neumann boundary condition $\psi'(0)=0$).

The familiar Trotter product formula further requires that $\tilde H = \tilde{T}_\beta + \tilde V$ be essentially self-adjoint on $D(\tilde T_\beta) \cap D(\tilde V)$, a problematic condition in the presence of multiple self-adjoint extensions of $\tilde H$ for $|J_+|<1$. The extended results of Kato and of Exner et al do not require essential self-adjointness of $\tilde{T}_\beta + \tilde V$, but instead demand that $\tilde{T}_\beta$ and $\tilde V$ be nonnegative. Consequently their application is limited to factor orderings such that $|J_+|\geq \frac{1}{2}$ and hence $V (y) = \hbar^2\big(|J_+|^2-\frac{1}{4}\big)\frac{1}{y^2}+\frac{\Lambda}{4}y^2 \geq 0$.

The following subsections construct Feynman--Ka\v{c} formulas in real and imaginary time, for the cases $|J_+| \geq \frac{1}{2}$. The case $|J_+|= \frac{1}{2}$ is special because for such factor orderings, the quantum potential term $\hbar q y^{-2} = \hbar \big(J_+^2 - \frac{1}{4}\big) y^{-2}$ in $V(y)$ vanishes. For $|J_+|< \frac{1}{2}$, neither the standard Trotter product formula nor extensions due to Kato and to Exner et al apply. In the case of the extensions, the operator $\tilde V$ fails to be nonnegative, while the standard Trotter product theorem does not apply since $D(\tilde{T}_\beta) \cap D(\tilde V)$ does not determine one of the self-adjoint extensions of $\tilde H$. Indeed, functions $\psi \in D(\tilde V)$ satisfy
\begin{gather}
\int_0^\infty y^{-4} \lvert \psi \rvert ^2 \dd y < \infty ,\label{domain_V}
\end{gather}
while elements of the self-adjoint extension of $\tilde H$ corresponding to the angle $\theta$ in $\Span{\varphi^{(1)},\varphi^{(2)}}$ behave asymptotically as \eqref{thetasymptot} near $y=0$. For $|J_+|<\frac{1}{2}$, the condition~\eqref{domain_V} can only be satisf\/ied when $C=0$, meaning that $D(\tilde{T}_\beta) \cap D(\tilde V)$ satisf\/ies the boundary conditions of all self-adjoint extensions of $\tilde H$, and thus $\tilde H$ is not essentially self-adjoint on $D(\tilde{T}_\beta) \cap D(\tilde V)$. In fact, \eqref{domain_V} can only be satisf\/ied when $C=0$ for all $|J_+|<1$ (where $|J_+| \neq \frac{1}{2}$, since at $|J_+|=\frac{1}{2}$ \eqref{domain_V} does not apply), so the standard Trotter product formula is inapplicable for all such orderings.

Complementary to the path integral approach to constructing a time evolution operator one may use spectral decomposition of the Hamiltonian to def\/ine propagators as in \cite{PatelRiveraMaitra}.

\subsection{Imaginary-time Feynman--Ka\v{c} formula}

\subsubsection[Factor orderings with $|J_+| > \frac{1}{2}$]{Factor orderings with $\boldsymbol{|J_+| > \frac{1}{2}}$}

For factor orderings in this range, both $\tilde{T}_\beta$ and $\tilde V$ are nonnegative self-adjoint operators. In addition, the intersection of their form domains $Q(\tilde{T}_\beta) \cap Q(\tilde V)$ (equivalently, $D(\tilde{T}_\beta ^{1/2}) \cap D(\tilde V ^{1/2})$) is dense (since it contains $C_0^\infty (\RR^+)$). Thus the operators $\tilde{T}_\beta$ and $\tilde V$ satisfy the hypotheses of the strong Kato--Trotter product formula (see, e.g.,~\cite[Theorem~S.21]{reedsimon1}). We conclude that the quadratic form $\langle \varphi , \tilde{T}_\beta \psi \rangle_{L^2(\RR^+)} + \langle \varphi , \tilde V \psi \rangle_{L^2(\RR^+)}$ is closed on the domain $Q(\tilde{T}_\beta ) \cap Q(\tilde V)$, and hence is associated with a self-adjoint operator, denoted $\tilde{T}_\beta \dot{+} \tilde V$ (the \emph{form sum} of $\tilde{T}_\beta$ and~$\tilde V$). The product formula
\begin{gather}
 e^{-t (\tilde{T}_\beta \dot{+} \tilde V )} =\slim_{n \to \infty} \big( e^{-t \tilde{T}_\beta / n} e^{-t \tilde V / n} \big)^n \label{trotter_product}
\end{gather}
then holds.

We f\/irst investigate the domain $D(\tilde{T}_\beta \dot{+} \tilde V ) \subseteq Q(\tilde{T}_\beta ) \cap Q(\tilde V)$ for the self-adjoint ope\-ra\-tor~$\tilde{T}_\beta \dot{+} \tilde V$. Observe that since
\begin{gather*}
Q(\tilde V) = \left\{ \psi \in L^2 (\RR^+ ) \, \big | \, \int_0^\infty V(y) \lvert \psi (y) \rvert^2 \dd y < \infty \right\} ,
\end{gather*}
every element $\psi \in D(\tilde{T}_\beta \dot{+} \tilde V)$ must satisfy
\begin{gather}
\int_0^\infty y^{-2} \lvert \psi(y) \rvert^2 \dd y < \infty .\label{form_domain_asymptotics}
\end{gather}

For $|J_+| \geq 1$, there is only one self-adjoint domain for $\tilde H$, and the asymptotic analysis of Section~\ref{Asymptotics} implies that the elements satisfy $\psi \sim O(y^{\frac{3}{2}}) $ near $y=0$, consistent with~\eqref{form_domain_asymptotics}; indeed, $D(\tilde{T}_\beta \dot{+} \tilde V)$ must coincide with the single self-adjoint extension of $\tilde H$ (for all values of $\beta$). For $\frac{1}{2} < |J_+| <1$, the self-adjoint extensions of $\tilde H$ are indexed by one-dimensional subspaces of $\Span{\varphi^{(1)} , \varphi^{(2)}}$, with each subspace determined by $\theta$ as in~\eqref{asymptotics_theta}, so that elements of $D(\tilde{T}_\beta \dot{+} \tilde V)$ must be characterized by a particular angle $\theta = \theta'$. Combining~\eqref{form_domain_asymptotics} and~\eqref{asymptotics_theta}, we obtain
\begin{gather}
\int_0^1 y^{-2} \big[\sin (\theta') \varphi^{(1)}(y) + \cos ( \theta') \varphi^{(2)} (y)\big]^2 \dd y < \infty .
\label{form_domain_restrict}
\end{gather}
Since $\varphi^{(1)} \sim \mathscr{C}_1 y^{-|J_+|+ \frac{1}{2}}$, $\varphi^{(2)} \sim \mathscr{C}_2 y^{|J_+|+\frac{1}{2}}$, \eqref{form_domain_restrict} can only be satisf\/ied if $\theta' = 0$, identifying $\tilde{T}_\beta \dot{+} \tilde V$ as the self-adjoint extension of $\tilde H$ purely in $\Span{\varphi^{(2)}}$ (for all $\beta$).

For $|J_+|>\frac{1}{2}$ the conditions of the Kato--Trotter product formula identify a single self-adjoint extension of $\tilde H$ for which one may construct a time-sliced path integral, regardless of the choice of~$\tilde{T}_\beta$. Thus without loss of generality we f\/ix $\tilde{T}_0$ (corresponding to the Dirichlet boundary condition on the Laplacian) as the self-adjoint extension for $\tilde T$.

We now combine the Kato--Trotter product formula with the well-known Dirichlet heat kernel on $\RR^+$ (see, e.g., \cite{ClarkMenikoffSharp, FarhiGutmann, SaloffCosteLecture}) to construct a Feynman--Ka\v{c} formula for the imaginary-time evolution operator $e^{-t \tilde{H}/ \hbar}$ (guaranteed by the spectral theorem to be well-def\/ined for self-adjoint $\tilde{H}$). The Dirichlet heat kernel on the half-line, given by
\begin{gather}
p_+^D(t,y,z) = \frac{1}{\sqrt{4\pi t}} \big[ e^{-(y-z)^2/4t} - e^{-(y+z)^2/4t} \big] ,\label{DirichletHeatKernel}
\end{gather}
is the integral kernel for $e^{t \Delta_0}$, and may be obtained by applying the ordinary heat kernel on $\RR$ to an odd extension of data $\psi: \RR^+ \to \CC$ (with $\psi(0) = 0$). The heat kernel~\eqref{DirichletHeatKernel} preserves the Dirichlet boundary condition for all time. From~\eqref{trotter_product} and~\eqref{DirichletHeatKernel}, we obtain
\begin{gather}
e^{-t \tilde H / \hbar} \psi = \lim_{n \to \infty} \big( e^{-t \tilde{T}_0 /\hbar n} e^{-tV(y) / \hbar n} \big)^n \psi = \lim_{n \to \infty} \big( e^{\hbar t \Delta/n} e^{-tV(y) / \hbar n} \big)^n \psi \nonumber\\
\hphantom{e^{-t \tilde H / \hbar} \psi}{} = \lim_{n \to \infty} \mathcal{N} \int_{\RR^+} \!\cdots\! \int_{\RR^+} \prod_{j=1}^n \Big[ \Big( e^{ \frac{-n (z_{j+1} - z_j)^2}{4 \hbar t}} - e^{\frac{-n (z_{j+1} + z_j)^2}{4 \hbar t}} \Big) e^{\frac{-tV(z_j)}{ \hbar n}} \Big] \psi(z_1) \dd z_1 \cdots \dd z_n , \nonumber\\
 \mathcal{N} = \left( \frac{n}{4 \pi \hbar t} \right)^{n/2} .\label{time_slice_Dirichlet}
\end{gather}

The path integrals corresponding to dif\/ferent factor orderings are distinguished by the quantum potential $\hbar^2 q y^{-2} = \hbar^2\big( J_+^2 - \frac{1}{4}\big) y^{-2}$, so that we may write the factor
\begin{gather*}
\prod_{j=1}^n e^{-tV(z_j)/ \hbar n} = \prod_{j=1}^n e^{-t \Lambda z_j^2 / 4 \hbar n} \prod_{k=1}^n e^{-t \hbar q z_k^{-2}/n}
\end{gather*}
and regard the latter product on the right-hand side as a weight on the path integral measure corresponding to the factor ordering chosen. Notice that this weight is small for histories tarrying near the singularity $y=0$. A large value of $q$ in the quantum potential intensif\/ies this ef\/fect.

The $n$-fold time-sliced measure
\begin{gather*}
\left( \frac{n}{4 \pi \hbar t} \right)^{n/2} \prod_{j=1}^n \Big( e^{\frac{-n (z_{j+1} - z_j)^2}{4 \hbar t}} - e^{\frac{-n (z_{j+1} + z_j)^2}{4 \hbar t}} \Big) \dd z_1 \cdots \dd z_n
\end{gather*}
represents an \emph{avoiding measure} supported on paths away from the singularity $y=0$. As discussed by Farhi and Gutmann in~\cite{FarhiGutmann}, the f\/irst term of~\eqref{DirichletHeatKernel} is the heat kernel for propagation from $z>0$ to $y>0$ on $\RR$, but to construct an avoiding measure on the half-line $\RR^+$, this must be adjusted by subtracting the weight for paths from $y>0$ to $z>0$ incident on the origin. The appropriate weight to subtract is precisely the second term of \eqref{DirichletHeatKernel}, because paths from $y>0$ to $z>0$ intersecting the origin are in one-to-one correspondence with paths from~$y$ to $-z$ (for a full treatment, refer to \cite[Section~2]{FarhiGutmann}.)

We conclude that by introducing a quantum potential, factor orderings in the range~$|J_+| \geq 1$ require that the path integral measure be supported on paths avoiding the singularity. A~larger value of~$q$ results in a decreased contribution by histories for which the circumference of the universe remains mostly small. For orderings with $\frac{1}{2} < |J_+| < 1$, the extended product formulae limit the construction of a time-sliced path integral to one whose measure is supported on paths avoiding the singularity.

\subsubsection[Factor orderings with $|J_+|=\frac{1}{2}$]{Factor orderings with $\boldsymbol{|J_+|=\frac{1}{2}}$}

As in the preceding section $\tilde{T}_\beta$ and $\tilde V$ are nonnegative self-adjoint operators with \smash{$Q(\tilde{T}_\beta) \cap Q(\tilde V)$} dense, so that once again the hypotheses of the strong Kato--Trotter product formula are sa\-tisf\/ied, for $\beta \in (-\infty , \infty]$. Thus $\tilde{T}_\beta \dot{+} \tilde V$ is self-adjoint on $D(\tilde{T}_\beta \dot{+} \tilde V) \subseteq Q(\tilde{T}_\beta) \cap Q(\tilde V)$, and the product formula~\eqref{trotter_product} holds.

We wish to identify $D(\tilde{T}_\beta \dot{+} \tilde V)$ as one of the self-adjoint extensions of $\tilde H$, corresponding to $\theta $ in $\Span{\varphi^{(1)} , \varphi^{(2)}}$. However, the absence of the quantum potential (since $q=J_+^2 - \frac{1}{4}=0$) removes the condition \eqref{form_domain_asymptotics}; membership of $\psi$ in $Q(\tilde V)$ no longer imposes any restriction on asymptotics near $y=0$, and instead we utilize that $D(\tilde{T}_\beta) \cap D(\tilde V) \subseteq D(\tilde{T}_\beta \dot{+} \tilde V)$ (see, e.g., \cite[Proposition~3.1]{Faris}). Thus the corresponding self-adjoint extension of $\tilde H$ contains elements satisfying the boundary condition $\psi(0)=\beta \psi'(0)$. Since the asymptotic analysis of Section~\ref{Asymptotics} yields (for $|J_+|=\frac{1}{2}$),
\begin{gather}
\psi \sim C \big( \mathscr{C}_1 \sin(\theta) + \mathscr{C}_2 \cos(\theta ) y \big) + O\big(y^{3/2}\big), \nonumber\\
\psi' \sim C \big( \mathscr{C}_2 \cos(\theta) \big) + O\big(y^{1/2}\big) ,\label{Jhalf_asymptotics}
\end{gather}
we conclude that $\beta$ determining the self-adjoint extension of the Laplacian is related to $\theta$ for the self-adjoint extension of $\tilde H$ by
\begin{gather}
\beta= \frac{\mathscr{C}_1}{\mathscr{C}_2} \tan \theta .\label{betatheta}
\end{gather}
The Dirichlet boundary condition $\beta = 0$ corresponds to $\theta = 0$, and as in the preceding section, we use~\eqref{DirichletHeatKernel} with the Kato--Trotter product formula to obtain the same time-sliced path integral expression for $e^{-t \tilde{H}/ \hbar}$ (cf.~\eqref{time_slice_Dirichlet}).

The Neumann boundary condition $\psi'(0)=0$ (formally $\beta = \infty$) corresponds to $\theta=\frac{\pi}{2}$. In this case $e^{-t \tilde{T}_\infty /n}$ in~\eqref{trotter_product} may be expressed in terms of the well-known Neumann heat kernel on~$\RR^+$ (see, e.g.,~\cite{ClarkMenikoffSharp, FarhiGutmann, GaveauSchulman, SaloffCosteLecture}):
\begin{gather}
p_+^N(t,y,z) = \frac{1}{\sqrt{4 \pi t}} \big[ e^{-(y-z)^2/4t} + e^{-(y+z)^2/4t} \big] ,\label{NeumannHeatKernel}
\end{gather}
derived by evenly extending data on $\RR^+$ (with $\psi'(0)=0$) to $\RR$ and evolving via the ordinary heat kernel. Similarly to the case for Dirichlet boundary conditions, we combine~\eqref{NeumannHeatKernel} with the Kato--Trotter product formula to obtain the time-sliced path integral propagator:
\begin{gather*}
e^{-t \tilde H / \hbar} \psi = \lim_{n \to \infty} \big( e^{-t \tilde{T}_\infty /\hbar n} e^{-tV(y) / \hbar n} \big)^n \psi = \lim_{n \to \infty} \big( e^{\hbar t \Delta/n} e^{-tV(y) / \hbar n} \big)^n \psi \\
\hphantom{e^{-t \tilde H / \hbar} \psi }{}
= \lim_{n \to \infty} \mathcal{N} \int_{\RR^+}\! \cdots\! \int_{\RR^+} \prod_{j=1}^n \Big[ \Big( e^{\frac{-n (z_{j+1} - z_j)^2}{4 \hbar t}} + e^{\frac{-n (z_{j+1} + z_j)^2}{4 \hbar t}} \Big) e^{\frac{-tV(z_j)}{\hbar n}} \Big] \psi(z_1) \dd z_1 \cdots \dd z_n , \\
\mathcal{N} = \left( \frac{n}{4 \pi \hbar t} \right)^{n/2} .
\end{gather*}

In contrast to the singularity-avoiding measure generated by the Dirichlet heat kernel, the Neumann $n$-fold time-sliced measure \begin{gather*}\left( \frac{n}{4 \pi \hbar t} \right)^{n/2} \prod_{j=1}^n \Big( e^{\frac{-n (z_{j+1} - z_j)^2}{4 \hbar t}} + e^{\frac{-n (z_{j+1} + z_j)^2}{4 \hbar t}} \Big) \dd z_1 \cdots \dd z_n\end{gather*} corresponds to a \emph{reflecting measure} supported on paths which may strike the singularity on multiple occasions. As discussed in~\cite{FarhiGutmann}, the f\/irst term in~\eqref{NeumannHeatKernel} yields the weight for paths striking the origin an even number of times, since these paths are in one-to-one correspondence with paths from $y>0$ to $z>0$ in the whole line. The second term in~\eqref{NeumannHeatKernel} supplements the f\/irst by weighting paths which strike the origin an odd number of times, since the same correspondence associates such paths with paths in the whole line from $y>0$ to $-z<0$.

We have derived path integral representations for the imaginary-time evolution operators corresponding to two of the self-adjoint extensions of $\tilde H$ ($\theta=0$ and $\theta=\frac{\pi}{2}$) when $|J_+|=\frac{1}{2}$, by using the heat kernel for Dirichlet and Neumann boundary conditions, respectively. For the remaining extensions, the heat kernel for the boundary condition $\psi(0)=\beta \psi'(0)$ may be constructed using spectral decomposition of~$-\Delta$ with corresponding boundary behavior. The resulting eigenfunctions are $\psi_p(y) = \big( \frac{2}{\pi} \big)^{1/2} \cos(py+ \phi_p)$, where $\tan \phi_p = -1/ p\beta$ (see \cite{ClarkMenikoffSharp,FarhiGutmann, GaveauSchulman}), yielding the heat kernel
\begin{gather*}
p_+^{(\beta)}(t,y,z) = \frac{1}{\sqrt{4 \pi t}} \big[ e^{-(y-z)^2/4t} + e^{-(y+z)^2/4t} \big] \\
\hphantom{p_+^{(\beta)}(t,y,z) = }{} - \beta^{-1} \exp \left\{ \frac{t}{\beta^2} + \frac{(y+z)}{\beta} \right\} \operatorname{erfc} \left[ \left( \frac{t}{\beta^2} \right)^{1/2} + \left( \frac{y+z}{2 \beta} \right) \left( \frac{t}{\beta^2} \right)^{-1/2} \right] .
\end{gather*}
The term containing the complementary error function $\operatorname{erfc}(\cdot)$ decreases as $\beta$ ranges from~$0$ to~$\infty$, and may be regarded as interpolating between the avoiding and ref\/lecting measures.

The time-sliced path integral expression for the imaginary-time evolution operator may thus be constructed as
\begin{gather*}
e^{-t \tilde H / \hbar} \psi = \lim_{n \to \infty} \int_{\RR^+} \cdots \int_{\RR^+} \prod_{j=1}^n \left[ p_+^{(\beta)} \left( \frac{\hbar t}{n} , z_{j+1} , z_j \right) e^{-tV(z_j)/ \hbar n} \right] \psi(z_1) \dd z_1 \cdots \dd z_n .
\end{gather*}

\subsection{Real-time Feynman--Ka\v{c} formula}

The Trotter product formula for real time is generally applicable to negative as well as non\-ne\-ga\-tive operators. However as noted above, when the operator $\tilde V$ is negative (i.e., when $|J_+|<\frac{1}{2}$), the intersection of the domains $D(\tilde{T}_\beta) \cap D(\tilde V)$ is too small to determine a unique self-adjoint extension of $\tilde{H}$, and the standard Trotter product formula does not apply.

When both $\tilde{T}_\beta$ and $\tilde V$ are nonnegative ($|J_+| \geq \frac{1}{2}$), a result due to Exner et al.\ (see \cite[Theorem~2.2]{ExnerNeidhardtZagrebnov}) allows us to conclude that a version of the Trotter product formula in real time holds with respect to the $L^2$ norm over time. That is, for nonnegative self-adjoint $\tilde{T}_\beta$ and $\tilde V$ with densely def\/ined form sum $\tilde{T}_\beta \dot{+} \tilde{V}$ and for any time $T>0$ and any state $\psi \in L^2(\RR^+)$,
\begin{gather}
\lim_{n \to \infty} \int_{-T}^T \Big\| \big( e^{-it \tilde{T}_\beta / \hbar n} e^{-it \tilde V / \hbar n} \big)^n \psi - e^{-it(\tilde{T}_\beta \dot{+} \tilde{V}) / \hbar } \psi \Big\|_2^2 \dd t = 0 .\label{exnerproduct}
\end{gather}
As discussed in \cite{ExnerNeidhardtZagrebnov}, this result is weaker than the standard Trotter product formula; however it implies the existence of a subsequence $\{n_k\}$ along which convergence in the strong operator topology holds. This is suf\/f\/icient to def\/ine a time-sliced path integral.

The conditions of the extended Trotter product formula \eqref{exnerproduct} are identical to those of the strong Kato--Trotter product formula. Hence an identical analysis determines the self-adjoint extension of $\tilde H$ admitting a time-sliced path integral expression. As before, for $|J_+| \geq 1$, $\tilde{T}_0 \dot{+} \tilde V$ coincides with the unique self-adjoint extension of $\tilde H$, while for $\frac{1}{2}<|J_+|<1$, the restriction~\eqref{form_domain_asymptotics} picks out the $\theta = 0$ extension (Dirichet boundary condition) as that to which~\eqref{exnerproduct} applies. For $|J_+|=\frac{1}{2}$, the correspondence \eqref{betatheta} equates $D(\tilde{T}_\beta \dot{+} \tilde V)$ to the self-adjoint domain for $\tilde H$ with asymptotics given by~\eqref{Jhalf_asymptotics}, and each extension admits a Feynman--Ka\v{c} formula for the real-time evolution operator.

For orderings satisfying $|J_+| \geq 1$ or $|J_+| = \frac{1}{2}$, we have that $D(\tilde{T}_\beta) \cap D(\tilde V)$ is large enough to determine a unique self-adjoint extension of $\tilde H$, implying that the standard Trotter product formula applies.

\subsubsection[Factor orderings with $|J_+| \geq 1$]{Factor orderings with $\boldsymbol{|J_+| \geq 1}$}
For factor orderings with $|J_+| \geq 1$, the self-adjoint form sum $\tilde{T}_0 \dot{+} \tilde V$ must coincide with the unique self-adjoint extension of $\tilde H$. At the same time, $ D(\tilde{T}_0) \cap D(\tilde V) \subseteq D(\tilde{T}_0 \dot{+} \tilde V)$ and $D(\tilde{T}_0) \cap D(\tilde V)$ is dense (since it contains $C_0^\infty (\RR^+)$), so that $\tilde H = \tilde{T}_0 + \tilde V$ is essentially self-adjoint on $D(\tilde{T}_0) \cap D(\tilde V)$. Thus the standard Trotter product formula applies, and together with the well-known integral kernel \cite{ClarkMenikoffSharp, FarhiGutmann, GaveauSchulman}
\begin{gather*}
g_+^D(t,y,z) = \frac{1}{\sqrt{4 \pi it}} \big[ e^{i(y-z)^2/4t} - e^{i(y+z)^2/4t} \big]
\end{gather*}
for $e^{it \Delta}$ permits a real-time Feynman--Ka\v{c} formula for $e^{-it \tilde{H} / \hbar}$:
\begin{gather}
e^{-it \tilde H / \hbar} \psi = \lim_{n \to \infty} \big( e^{-it \tilde{T}_0 /\hbar n} e^{-itV(y) / \hbar n} \big)^n \psi = \lim_{n \to \infty} \big( e^{i\hbar t \Delta/n} e^{-itV(y) / \hbar n} \big)^n \psi \nonumber\\
\hphantom{e^{-it \tilde H / \hbar} \psi}{} = \lim_{n \to \infty} \mathcal{N} \int_{\RR^+} \!\cdots\! \int_{\RR^+} \prod_{j=1}^n \Big[ \Big( e^{\frac{in (z_{j+1} - z_j)^2}{4 \hbar t}} - e^{\frac{in (z_{j+1} + z_j)^2}{4 \hbar t}} \Big) e^{\frac{-itV(z_j)}{ \hbar n}} \Big] \psi(z_1) \dd z_1 \cdots \dd z_n , \nonumber\\
\mathcal{N} = \left( \frac{n}{4 \pi i \hbar t} \right)^{n/2} .\label{real_time_slice_Dirichlet}
\end{gather}
As in \eqref{time_slice_Dirichlet}, the $n$-fold measure appearing in \eqref{real_time_slice_Dirichlet} corresponds to the avoiding measure supported on paths which do not intersect the singularity.

\subsubsection[Factor orderings with $\frac{1}{2}<|J_+|<1$]{Factor orderings with $\boldsymbol{\frac{1}{2}<|J_+|<1}$}

As discussed above, in this range of orderings the ordinary Trotter product formula does not apply, but the result~\eqref{exnerproduct} of~\cite{ExnerNeidhardtZagrebnov} permits a Feynman--Ka\v{c} formula. As in the imaginary-time case, the identif\/ication of $D(\tilde{T}_\beta \dot{+} \tilde V)$ with the Dirichlet self-adjoint domain of $\tilde H$ ($\theta =0$), independent of $\beta$, implies that without loss of generality, $\tilde{T}_0$ may be taken as the kinetic term. The Feynman--Ka\v{c} formula will thus be as in~\eqref{real_time_slice_Dirichlet}, with the limit understood to be taken over the subsequence~$\{ n_k \}$ for which strong operator convergence of the product formula is guaranteed.

\subsubsection[Factor orderings with $|J_+| = \frac{1}{2}$]{Factor orderings with $\boldsymbol{|J_+| = \frac{1}{2}}$}
As discussed in the imaginary-time case, for $|J_+| = \frac{1}{2}$ the intersection $D(\tilde{T}_\beta) \cap D(\tilde V)$ is contained in a unique self-adjoint extension of~$\tilde{H}$, with the correspondence between $\beta$ and the angle $\theta$ of the self-adjoint extension given by~\eqref{betatheta}. Since $C_0^\infty(\RR^+) \subset D(\tilde{T}_\beta) \cap D(\tilde V)$ is dense, $\tilde H$ is essentially self-adjoint on~$D(\tilde{T}_\beta) \cap D(\tilde V)$ and the standard Trotter product formula again applies.

The integral kernel for $e^{it\Delta}$ with boundary condition $\psi(0) = \beta \psi'(0)$ is given by \cite{ClarkMenikoffSharp, FarhiGutmann, GaveauSchulman}
\begin{gather}
 g_+^{(\beta)} (t,y,z) = \frac{1}{\sqrt{4 \pi i t}} \big[ e^{i(y-z)^2/4t} + e^{i(y+z)^2/4t} \big] \nonumber\\
\hphantom{g_+^{(\beta)} (t,y,z) =}{}
- \beta^{-1} \exp \left\{ \frac{it}{\beta^2} + \frac{(y+z)}{\beta} \right\} \operatorname{erfc} \left[ \left( \frac{it}{\beta^2} \right)^{1/2} + \left( \frac{y+z}{2 \beta} \right) \left( \frac{it}{\beta^2} \right)^{-1/2} \right] .\label{realtime_kernel}
\end{gather}
Using \eqref{realtime_kernel}, the time-sliced path integral expression for the real-time evolution operator becomes
\begin{gather*}
e^{-it \tilde H / \hbar} \psi = \lim_{n \to \infty} \int_{\RR^+} \cdots \int_{\RR^+} \prod_{j=1}^n \left[ g_+^{(\beta)} \left( \frac{\hbar t}{n} , z_{j+1} , z_j \right) e^{-itV(z_j)/ \hbar n} \right] \psi(z_1) \dd z_1 \cdots \dd z_n .
\end{gather*}
As with imaginary time, the special cases $\beta=0$ and $\beta=\infty$ correspond to the Dirichlet and Neumann boundary conditions, with
\begin{gather*}
g_+^{(\beta)} (t,y,z) = \begin{cases} \dfrac{1}{\sqrt{4 \pi it}} \big[ e^{i(y-z)^2/4t} - e^{i(y+z)^2/4t} \big] , & \beta = 0, \vspace{1mm}\\
\dfrac{1}{\sqrt{4 \pi it}} \big[ e^{i(y-z)^2/4t} + e^{i(y+z)^2/4t} \big] , & \beta = \infty .
\end{cases}
\end{gather*}
The interpretation of the path integral measures as supported on avoiding and ref\/lecting paths for the Dirichlet and Neumann cases, respectively, and of the $\operatorname{erfc}(\cdot)$ term in the integral ker\-nel~\eqref{realtime_kernel} as interpolating between these two cases, remains as in imaginary time.

\section{Conclusion}

This investigation highlights the ef\/fects of factor ordering in circumscribing the range of allowed behaviors for wave functions near the Big Bang/Big Crunch singularity. By generating a suf\/f\/iciently steep positive quantum potential, orderings with $|J_+| \geq 1$ force the correspon\-ding path integral measure to be supported only on paths which avoid the singularity, while orderings with $|J_+| =\frac{1}{2}$ (including the Laplace--Beltrami/conformal ordering), generating no quantum potential, require an additional boundary condition on wave functions to determine a~path integral measure. This boundary condition dictates whether the path integral measure is supported solely on singularity-avoiding paths or whether it has support on paths which ref\/lect of\/f of the singularity.

These results af\/f\/irm the observation of Kontoleon and Wiltshire \cite{KontoleonWiltshire} regarding the sensitivity of boundary behavior in quantum cosmology to factor ordering. Moreover, the importance of boundary behavior in addressing the question of quantum resolution of the classical Big Bang/Big Crunch singularity has been discussed by Kiefer in \cite{AnnalKiefer2010}; the results in Section~\ref{FK} ref\/lect the close connection between the support of the path integral measure and the behavior of wave functions near the singularity.

In the sequel we aim to investigate the construction of Hartle--Hawking-type wave functions of the universe using the path integral measures above, and solving the Wheeler--DeWitt equation with corresponding factor ordering. To address the inherent time-reparametrization invariance of (1+1) gravity, it would be of interest to construct time-independent propagators between initial and f\/inal states as by Bunster in \cite{Teitelboim1, Teitelboim2}; in \cite{Halliwell1988}, Halliwell implements such methods to yield solutions and Green's functions for the Wheeler--DeWitt equation in quantum cosmology. Additionally, one could implement the promising methods of Gerhardt in~\cite{Gerhardt} to characterize the space of physical states.

Eventually we expect to broaden the class of factor orderings considered, and ultimately to extend analysis to cosmological models such as Friedmann--Lema\^itre--Robertson--Walker and Bianchi spacetimes.

\subsection*{Acknowledgements}The authors are grateful to Renate Loll for many fruitful conversations which informed this project, to Jan Ambj{\o}rn for the reference to Nakayama's work on 2D quantum gravity, and to Vincent Moncrief, Antonella Marini, and Maria Gordina for ongoing useful discussions about quantization and mathematical physics. Additionally, the authors would like to thank the editor and the anonymous reviewers for their thoughtful and insightful comments and suggestions.

\pdfbookmark[1]{References}{ref}
\LastPageEnding

\end{document}